\documentclass[12pt,english]{article}
\usepackage{lmodern}
\usepackage[T1]{fontenc}
\usepackage[utf8]{inputenc}
\usepackage[a4paper,margin=2.5cm]{geometry}
\usepackage{color}
\usepackage{babel}
\usepackage{float}
\usepackage{mathtools}
\usepackage{amsmath}
\usepackage{amssymb}
\usepackage[unicode=true,
 bookmarks=false,
 breaklinks=false,pdfborder={0 0 1},
 colorlinks=true]
 {hyperref}
\hypersetup{
 urlcolor=refcolor,citecolor=refcolor}

\makeatletter
\usepackage{amsfonts}
\usepackage{amsthm}
\usepackage{tablefootnote}
\usepackage{xcolor}
\usepackage{url}
\usepackage{tabularx}
\usepackage{graphicx}
\usepackage{caption}
\usepackage{tikz}
\usepackage{pgfplots}
\pgfplotsset{compat=newest}
\usepgfplotslibrary{fillbetween}

\usepackage{csquotes}
\usepackage[round]{natbib}
\usepackage{float}
\usepackage{tocbibind}

\definecolor{refcolor}{rgb}{0, 0, 0.5}

\newtheorem{proposition}{Proposition}\theoremstyle{definition}
\newtheorem{assumption}{Assumption}

\makeatother

\date{{\large 18 July, 2023}}

\begin{document}
\title{Should Politicians be Informed? \\ Targeted Benefits and Heterogeneous Voters\thanks{We would like to thank Nuray Akin, Kemal Kivanç Aköz, Nageeb Ali, Daniel Treisman, Onursal Bağırgan, Nicholas Ziros, Jan Z\'{a}pal, Dmitriy Vorobyev, Ole Jann, Artyom Jelnov, Yiman Sun, Vladimir Shchukin, and the participants of the Workshop on Lobbying and Political Influence at the University of Utrecht for their helpful comments.}}
\author{
        Maxim Senkov\thanks{European Research University, U Haldy 200/18, 700 30 Ostrava, Czech Republic. email: maxim.senkov@eruni.org.} \\
        \and
	Arseniy Samsonov\thanks{Özyeğin University, Çekmeköy Campus Nişantepe District, Orman Street, 34794 Çekmeköy - Istanbul, Türkiye.
 email: arseniy.samsonov@ozyegin.edu.tr.}}
%\author[1]{Arseniy Samsonov\footnote{email address: arseniy.samsonov@ozyegin.edu.tr}}
%\author[2]{Maxim Senkov\footnote{email address: maxim.senkov@prigo.cz}}
%\affil[1]{\small{Ozyegin University}}
%\affil[2]{European Research University}
\date{\today}
\maketitle
\begin{abstract}
{
We compare two scenarios in a model where politicians offer local public goods to heterogeneous voters: one where politicians have access to data on voters and thus can target specific ones, and another where politicians only decide on the level of spending. When the budget is small, or the public good has a high value, access to voter information leads the winner to focus on poorer voters, enhancing voter welfare. With a larger budget or less crucial public goods, politicians target a narrow group of swing voters, which harms the voter welfare.
}\\
 {\footnotesize{}{}{}{}{}{} }\textbf{\footnotesize{}{}{}{}{}{}Keywords:}{\footnotesize{}{}{}{}{}{}
microtargeting; distributive politics; elections}\\
 {\footnotesize{}{}{}{}{}{} }\textbf{\footnotesize{}{}{}{}{}{}JEL
Classification Numbers:}{\footnotesize{}{}{}{}{}{} D82, D83.}{\footnotesize\par}

\end{abstract}

\newpage
\section{Introduction}
Do voters benefit when politicians get their data? A  benevolent policymaker designs better policies if she is more informed. However, the impact is less obvious when office-motivated competing politicians learn more about voters. In the US, politicians can design electoral promises based on individual data. At the same time, in Europe, there is an active debate on whether politicians should have this opportunity. Recently, European lawmakers restricted politicians' use of voter information \citep{brusselstimes2023}. 

We contribute to this discussion with a simple model where an Incumbent and a Challenger compete by offering a local public good to voters. Voters differ in their value of the public good, which perfectly correlates with their ideological position, reflecting the left-right spectrum. Politicians have a fixed budget to spend. We compare two scenarios: one in which politicians are informed and can target specific voters and a scenario in which they can choose the amount of spending, but not the precise voters who will get it. 

First, consider the "informed politicians" scenario. When the budget is small, or the importance of the public good to voters is high, the Incumbent wins by depleting the Challenger's resources. To do so, she exhausts the entire budget by targeting the Challenger's supporters. To secure enough votes for victory, the Challenger has to spend more than the budget allows, resulting in victory for the Incumbent. Because the public good goes to the Challenger's relatively poor supporters, it bestows a high benefit to the society overall.  In contrast, when politicians are uninformed, public goods often go to voters who do not strongly need them, leading to lower social welfare. Therefore, in the considered case, informed politicians do more good to society than uninformed ones, for the same reason that an informed and benevolent policymaker does.

With informed politicians and either a large budget or low importance of the public good to voters, the Challenger targets all Incumbent's swing voters, who constitute a relatively small group. To avoid defeat, the Incumbent also offers the public good to all voters in this group. Doing so allows the Incumbent to secure the votes of all her supporters and win. As a result, only the Incumbent's swing voters get the public good. If politicians are uninformed, the Incumbent cannot make a specific promise to this group of voters. Politicians must promise to allocate substantial portions of the budget to match each other's offers. When the budget is large or the importance of the public good to voters is low, this effect causes the scenario with uninformed politicians to yield higher voter welfare. 

To illustrate the case of the small budget, consider the 2002 gubernatorial campaign in Massachusetts, US. On the one hand, the state's budget was very tight at the time \citep{NPR2011RomneyBudget}. On the other hand, US politicians are known to rely heavily on voter polling, and, more recently, on data that social media generates. 
%\citep{hamel2018you}. 
One of the main contenders for the gubernatorial position, Mitt Romney, learned that he was not popular among women compared to his female competitor, Shannon O'Brien \citep{Freedlander2023}. Women's concern with healthcare topics \citep{Shuppy2008} may have led to Romney's campaign proposal to expand medical coverage \citep{Ebbert2002}. Given the state's budget problems at the time, O'Brien could have hardly retained enough voters by offering generous spending, which she did not \citep{Freedlander2023}. After winning the election, Romney adopted a state-level reform that American politicians and analysts came to view as Obamacare's successful precursor \citep{Thrush2012}. Thus, in this example, the information-driven political competition resulted in the implementation of policies that improved welfare.   

The effect of information when the budget is relatively large can be illustrated using federal-level examples from developed economies such as Belgium and Germany. The Belgian Vlaams Belang party may be an example of how politicians' knowledge of voter information may coincide with a lack of welfare-improving policy proposals. The Flemish nationalists successfully use individual data to identify potential supporters and then target them with benefit offers \citep{brusselstimes2023,Belga2022FlemishParty}. While focusing on immigration and ethnic identity, the party tries to get extra votes by adopting an economic program. However, its economic policy is not well-defined and incoherent, ranging from advocating deregulation favoring small business \citep{coffe2008small} to demanding more social spending \citep{Chini2022FlemishVlaamsBelang}. During the COVID-19 pandemic, the party advocated financial support for small business owners without clear evidence that this group needed help the most \citep{Sijstermans2021VlaamsBelangA}.  
In contrast, despite not getting detailed voter information \citep{kruschinski2017restrictions}, German politicians design policies that largely meet challenges such as poverty, single parenting, the pandemic, and setbacks in education \citep{SGI2020GermanySocial}. Importantly, as German politicians do not get fine-grained data on voters' needs, instead of tailoring campaigns to particular groups \citep{khazan2013}, they commit to policies benefiting broad categories of voters. 

The remainder of the paper is organized as follows: we first review the related literature (Section \ref{sect-liter}), then we present a simple model of political competition through public good provision promises (Section \ref{sect-model}), solve the model and present the main result of the welfare comparison between scenarios with informed and uninformed politicians (Section \ref{sect-analysis}), and, finally, we conclude (Section \ref{sect-concl}). In Appendix \ref{app-extension}, we consider an extension of the baseline model, supporting the welfare comparison result in Section \ref{sect-analysis}.

%In 2018, the campaign strategists of Greg Abbott learned that suburban women were frustrated by red-light cameras, and the candidate promised this group of voters to remove them. After becoming the governor of Texas, Abbott carried out his promise \citep{latimes2020}. Such a campaign would not have been possible without microtargeting, a technique that allows gathering fine-grained voter data and tailoring campaign promises to narrow groups of citizens. Is microtargeting good or bad for voters? 

\section{Literature} \label{sect-liter}

The project contributes to the rich literature on modeling distributive politics.\footnote{Most of this literature features politicians attempting to buy votes by promising redistribution. A separate research direction is studying how voters buy votes from each other, see e.g. \cite{tsakas2021vote}.} The paper contributes to the literature by exploring the new question concerning the impact of politicians’ knowledge of voter information on public good provision to the heterogeneous voters, and on voters' welfare.
The closest papers are \cite{myerson1993incentives}, \cite{lizzeri2001provision}, and \cite{dixit1996determinants}. The first two papers analyze politicians who compete by redistributing a fixed budget among homogeneous voters. Like this paper, \cite{myerson1993incentives} is concerned with politicians providing benefits to small groups. Unlike in our model, such behavior does not arise in a winner-take-all system with two politicians and is limited to more complicated electoral rules with multiple ones. As voters are homogeneous and politicians must spend the whole budget, no inefficiency arises. 

\cite{lizzeri2001provision} allow each politician to spend the whole budget on the public good that benefits all voters. The public good provision is more efficient than redistribution. In equilibrium, politicians buy the majority of voters using redistribution with positive probability, which leads to an inefficient outcome. Like in \cite{myerson1993incentives}, voters are homogeneous, and politicians must spend the whole budget. Hence, unlike this paper, \cite{lizzeri2001provision} do not consider inefficiencies that arise from politicians spending too little or targeting citizens who do not strongly need benefits. 
%There are many models in this area that study different aspects of the topic. For example, \citet{dixit1996determinants} analyze the politicians' choices of targeting core or swing voters. \citet{maskin2019pandering} model the impact of government transparency on fiscal discipline.  To our knowledge, this paper is the first to analyze how better access to information about voters that new technologies offer affects politicians' spending decisions and social welfare. 

In contrast to \cite{myerson1993incentives} and \cite{lizzeri2001provision}, \cite{dixit1996determinants} study the setting with voters who are heterogeneous in their ideological affinity to the candidates and in how they trade-off their political preferences versus economic benefits that candidates promise. This brings their setting close to the one we consider. However, \cite{dixit1996determinants} focus on identifying which voters the candidates promise economic benefits to. In line with the predictions of the "informed politicians" scenario in our paper, they predict that politicians tend to promise benefits to relatively more moderate voters.  In contrast to their paper, we find that when the budget is sufficiently large, the politicians tend to promise benefits to the swing voters who are relatively richer, which harms voter welfare. Finally, neither of the three aforementioned papers compares a scenario when politicians can target specific voters to a benchmark such that they cannot.\footnote{\cite{krasa2014social} consider voters, who are heterogeneous in terms of ideology and wealth. They study how politicians decide on public good provision when voters balance their ideological preferences against their preference for public good provision. In contrast to our model, they do not allow the politicians to target benefits to specific subsets of voters and focus on comparative statics of equilibrium tax rate proposals.} 

Also, the paper contributes to the field of Public Economics. In particular, it is close to \cite{akbarpour2023redistributive}, who consider a benevolent social planner allocating a public good to heterogeneous citizens. The benevolent social planner cares about providing the public good to poor consumers and can not base the provision on willingness to pay. In this setup, extra information helps the social planner achieve her goal. A key distinction of our approach is that we consider a similar setup in the face of electoral competition by office-motivated candidates. At the same time, we abstract away from the mechanism design problem of learning citizens' types that \cite{akbarpour2023redistributive} focus on. In our setting, access to voter information might incentivize politicians to allocate the public good to those who need it most. However, it may also happen that the Incumbent identifies her core supporters who react to benefits, and thus, she spends less. Hence, our model shows that giving politicians more information about voters can increase or decrease voters' welfare. 

Our paper can also be related to the new and small strand of literature that studies the impact of political microtargeting on social welfare. To address this issue, several recent papers look at politicians' ability to communicate differently with different voters by accessing their data and using social media.  For example, \cite{prummer2020micro} relates the media structure to the polarization that arises in the face of microtargeting. In \citet{titova2022targeted}, the challenger wins the election by providing different verifiable messages about her policy to different voters. \citet{titova2022targeted} finds that microtargeting reduces welfare because the challenger provides less accurate information to voters. Instead of considering which voters politicians choose to give information to, we look at who the politicians promise public goods to and show that letting politicians make promises to specific voter groups might be welfare-improving.
%The contrast between her results and ours shows that the impact of microtargeting depends on whether politicians use individual voter characteristics to design policies or simply communicate their intentions.

Several other papers are relevant because of their modeling assumptions or substantive topics. \cite{gregory2011rational} model a rational dictator who eliminates his constituents who might oppose his rule. Because he has limited information, he eliminates loyalists as well as dissidents. Better information leads to fewer eliminations. A somewhat similar effect appears in our model when uninformed politicians accidentally offer the public good to voters who put little value on them.

\section{Model} \label{sect-model}

\subsection{Motivation}

We analyze two scenarios, one in which politicians can perfectly target the public good provision and one in which they decide on its amount and random voters get access to it. We label one scenario as "informed politicians" and the second one as "uninformed politicians." Voters are distributed according to their valuation of the public good, which coincides with their ideological attitude to the politicians.

We make several strong assumptions to make the model simple and highlight the main trade-offs related to providing politicians with voter information. First, we assume that a voter's position on the $[0,1]$ segment simultaneously determines her value of the public good and ideological affinities to the politicians. To give an example, in the US, poorer citizens are more likely to identify as Democrats than as Republicans \citep{PewResearch2023}. They are also more likely to use state-sponsored goods such as public transport \citep{vox2015}.

Second, we assume that ties in voting are broken in favor of the Incumbent, i.e. it suffices for the Incumbent to obtain $\frac{1}{2}$ of votes to win the election. 
This Assumption is introduced primarily for the sake of technical convenience as it helps to
ensure the existence of Nash equilibrium. Essentially, this assumption can be rationalized as a representation of the electoral advantage held by the Incumbent.

Finally, we assume that some voters get the public good and others do not, and only those who do benefit from it. We disregard any spillover effects that might arise from the public good, as this drastically simplifies the model and analysis. This assumption is grounded in reality, as often public goods such as schools, hospitals, and bomb shelters\footnote{Israel is an example of a democratic country where bomb shelters are a necessity but not everyone has access to them \citep{timesofisrael2023}.} \emph{have a local scope} and capacity constraints, making them partially excludable. They bring more benefits to people living close to them. For instance, going to an emergency room might help only if it does not take too much time. Of course, people living far from a facility may still use it. However, our setup is a tractable model of how politicians distribute scarce public goods among different communities. 

Both "informed politicians" and "uninformed politicians" scenarios in our model are stylized. Consider the case when politicians are fully informed. Even though sometimes politicians can access individual voter data, it is hard to think of public goods specific to a single person. However, this setup describes a scenario where a politician can learn that a narrow group of citizens would benefit from a public good and promise it to them. 

On the other extreme, we consider a setup where politicians have no information about voters. This assumption may seem implausible because, in most countries, there is residential segregation by income, and in developed ones, governments collect reliable income data. They also understand that poorer people benefit more strongly from public goods. However, even within the low-income category, there may be economic or personal differences that lead to different outcomes of public goods provision and that are hard to capture. Take the example of public housing. For very poor people with a high risk of unemployment and homelessness, it is a life-saving option. However, as \cite{chetty2016effects} show, growing up in public housing decreases the future earnings of children by putting them in a worse social environment. Hence, while public housing might save a very poor family, it could negatively affect a slightly wealthier one who could afford to live in a better neighborhood. Without detailed data, policymakers would not be able to differentiate such families. Instead, they would provide public housing to a broad category of "low-income families." Like in our model, it would go to both families who do and do not benefit from it.

\subsection{Informed politicians} 

There is an Incumbent ($I$, she), a Challenger ($C$, he), and a set of voters $N = [0,1]$. For the sake of style, we may refer to either the Incumbent or the Challenger as a "politician." Each voter is characterized by her ideological position $t$. One can think of voters with low values of $t$ as \emph{low-income, left-leaning citizens} and those with higher values of $t$ as \emph{high-income and right-leaning}. 

The game proceeds as follows. First, the Incumbent and the Challenger simultaneously choose sets of voters $S_I, S_C \subseteq N$ to whom they promise to deliver public goods if elected. Next, each voter observes the politicians' choices and votes for the Incumbent or the Challenger. Because each voter is infinitesimal, she does not affect the outcome. However, we assume that she follows a heuristic and votes for the politician whose electoral promise gives her the highest expected payoff. We assume that if the expected payoffs are the same, a voter votes for each politician with probability $\frac{1}{2}$. Next, elections happen. The Incumbent wins if the measure of voters who vote for her is greater than or equal to $\frac{1}{2}$. Otherwise, the Challenger wins. Tie-breaking in favor of the Incumbent reflects an Incumbency advantage. At the last stage, payoffs are realized. 

Following \cite{de2005logic}, we assume that a politician can make herself more secure in office by retaining as much of the budget as possible. In developed democracies that this paper focuses on, a candidate running for office may need more resources in the future to reward loyal special interest groups, which is not useful for voters. Therefore, it benefits her if she can win while promising fewer expenditures. Naturally, politicians care about getting to power. To formalize these considerations, we assume that there is a fixed budget $v\in (0,1]$ and for politician $i \in \{I,C\}$, the payoff equals 

\begin{equation*}
%\label{p}
\begin{cases}
		 v - \mu(S_i), & \text{if $i$ wins}\\
           0, &  \text{otherwise,}
		 \end{cases}
\end{equation*}
where $\mu(S_i) \equiv \int_{S_i} dt$ is the measure of the set $S_i$. It captures the expenditures of politician $i \in \{I,C\}$.
%\footnote{\textcolor{magenta}{The Lebesgue measure is a mathematical concept that formalizes and generalizes ideas such as "length", "area", and "volume."} We show that in equilibrium, the sets $S_I$ and $S_C$ are intervals, and an interval's measure is simply its length. In general, calculating the Lebesgue measure of a set may be unintuitive, and for some sets, it is not defined. We use the concept of the measure to make the model complete. }. 
For simplicity, we restrict $S_I$ and $S_C$ to be finite collections of closed, non-overlapping intervals or empty sets. Thus, the measure $\mu(S_i)$ simplifies to the sum of the lengths of the intervals constituting $S_i$.

Each voter's payoff consists of an ideology component and a component related to public goods. As is standard in the political economy literature, we assume that a voter's ideological payoff from a politician taking power is the negative of the distance between their ideological positions. For tractability, we assume that the Challenger's ideological position is $0$ and the Incumbent's is $1$. Hence, if the Challenger wins, a voter with a position $t$ gets an ideological payoff of $-t$, and if the Incumbent wins, her ideological payoff equals $-(1-t)$. If a voter with position $t$ receives a public good, she gets an additional payoff of $\alpha (1-t)$. The parameter $\alpha > 0$ measures the importance of the public good to all voters, and the term $1 - t$ captures the effect that poorer voters put more value on public goods. In sum, for a voter with position $t$, the payoff equals

\begin{equation}
\label{voter_payoff_informed}
\begin{cases}
		-(1-t) + \alpha (1 - t) \cdot 1(t \in S_I), & \text{if the Incumbent wins}\\
         -t + \alpha (1 - t) \cdot 1(t \in S_C) , & \text{if the Challenger wins.}
		 \end{cases}
\end{equation}

\subsection{Uninformed politicians} \label{sect-uninfo-pol}

The politicians and voters are the same as in the previous environment. The main difference is that now politicians cannot promise public goods to specific voters. The game proceeds as follows. First, the Incumbent and the Challenger simultaneously choose $s_I$ and $s_C$, where $s_i, i \in \{I,C\},$ represents the share of voters who will get a public good. Unlike in the previous scenario, each voter, irrespective of her position $t$, will be able to receive a public good from politician $i \in \{I,C\}$ with probability $s_i$. Second, each voter observes whether she has been randomly drawn to be promised the public good by either the Incumbent or the Challenger, or by no one.  Next, each voter votes for the politician whose electoral promise gives her a higher expected payoff. We assume that if the expected payoffs are the same, a voter votes for each politician with probability $\frac{1}{2}$. The Incumbent wins if a measure of voters of at least $\frac{1}{2}$ votes for her. Otherwise, the Challenger wins. 

Finally, payoffs are realized. The payoff for politician $i \in \{I,C\}$ equals 
\begin{equation}
\label{p}
\begin{cases}
		 v - s_i, & \text{if $i$ wins}\\
          0, &  \text{otherwise}.
		 \end{cases}
\end{equation}

The mechanism through which voters are randomly drawn to either be promised the public good or not is worth a separate discussion.
For the sake of tractability, we follow \citet{myerson1993incentives} and assume that the politicians' promises create four sets of voters. The first set is those who will benefit only from the Incumbent if she wins, the second one is those who will benefit only from the Challenger, the third is those who will benefit from both, and the fourth one is those who will benefit from none. In each set, voters are uniformly distributed on the $[0,1]$ segment according to their type $t$. The sets have respective masses of $s_I (1-s_C)$, $(1-s_I)s_C$, $s_I s_C$, and $(1-s_I)(1-s_C)$. After politicians make promises, each voter learns the set to which she belongs. Each voter gets the same ideological payoff as she does in the previous environment. To summarize, with abuse of notation, define the set $S_I$ ($S_C$) as the voters who were randomly drawn to get benefits from the Incumbent (Challenger) in case she wins. Then the payoff for voter $t$ coincides with (\ref{voter_payoff_informed}). 

In the "uninformed politicians" scenario, a politician can not choose who of the voters gets access to the public good; meanwhile, given a promise, the voters know if they will be able to access the public good or not. A natural interpretation of this relates to local public goods provision. Imagine a politician who is uninformed regarding the distribution of wealth across neighborhoods. She announces a promise to construct a public hospital of a specific capacity in a given neighborhood. As the voters know both their location and the promised location of the hospital, they know if they would have access to the hospital or not, and they vote accordingly. In Appendix \ref{app-extension}, we consider an extension of this model,  where neither the politicians nor the voters know in advance who will have access to the public good ex-post. This models a scenario in which the politicians announce only the scope of public good provision; however, who among the voters gets access to the public good is determined only at the time of provision. We show that the main result of the paper (Proposition \ref{welfare}) still holds qualitatively.

\subsection{Social welfare of voters}

We are going to use the following definition to model how good or bad it is for society when politicians get voter information. Define the \textit{"social welfare of voters"} as
\begin{equation}
\label{voter_payoff_informed_new}
\begin{cases}
		\int_0^1 -(1-t) + \alpha (1 - t) \cdot 1(t \in S_I) dt, & \text{if the Incumbent wins}\\
       \int_0^1  -t + \alpha (1 - t) \cdot 1(t \in S_C) dt , & \text{if the Challenger wins}.
		 \end{cases}
\end{equation}
The social welfare of voters is simply the integral of voters' payoffs over the set of voters. 

\subsection{Assumptions}

%We make the following assumption about sets of voters to whom politicians offer public goods for the sake of tractability. 
%\begin{assumption}
%The sets $S_I$ and $S_C$ are Lebesgue measurable.
%\end{assumption}

For any positive value of public goods $\alpha$, there will be voters who vote for the Incumbent even if only the Challenger offers them public goods. We assume that there are also voters who vote for the Challenger even if only the Incumbent offers them public goods. It is natural to think of such voters are "partisans". Assuming that they exist simplifies the analysis and is realistic \citep{hillygus2008persuadable}. To guarantee that they do, it is necessary and sufficient that 
\begin{assumption}
$\alpha < 1$.
\end{assumption}
In our model, the only instrument the politicians can use to buy the voters' favor is the promise of public goods provision. Further, the voters who are ideologically closer to the Incumbent rather than to the Challenger (i.e., ones with $t>\frac{1}{2}$) get relatively smaller value from the public good. As a result, the promise of public good provision is relatively more effective in obtaining the votes of the Challenger's ideological supporters rather than the Incumbent's ideological supporters. Thus, the Incumbent's partisan supporters, given by $t\in\left[\frac{1+\alpha}{2+\alpha}, 1\right]$, are more numerous than the Challenger's partisan supporters given by  $t\in \left[0,\frac{1-\alpha}{2-\alpha}\right]$. The economic intuition of this model element is that poorer voters care more about economic benefits than political ideology.

\subsection{Solution concept}

The solution concept is the Nash equilibrium ("equilibrium" in the rest of the text). Because voters follow a heuristic, we do not treat them as players. Hence, the only players are the Incumbent and the Challenger. Since they move simultaneously, the choice of Nash equilibrium as the solution concept is natural.  

\section{Analysis} \label{sect-analysis}

We start with analyzing the "informed politicians" scenario. The equilibrium in this scenario is presented in Proposition \ref{informed}.

\begin{proposition}
\label{informed}
Suppose that politicians are informed. Then, an equilibrium always exists, and all equilibria have the following properties. 

If $v <  \frac{1 + \alpha}{2 + \alpha} - \frac{1}{2}$ (the budget is small or the value of the public good is large), the Incumbent offers the public good to a set of voters $S_I \subset [\frac{1 - \alpha}{2 - \alpha},\frac{1}{2}]$ with $\mu(S_I) = v$ and the Challenger offers the public good to $S_C \subset [\frac{1}{2}, \frac{1 + \alpha}{2 + \alpha}]$ with $\mu(S_C) = v$. Voters with positions $t \in \big((0,\frac{1}{2}) \setminus S_I \big) \cup S_C$ vote for the Challenger and voters with positions  $t \in \big((\frac{1}{2},1) \setminus S_C \big) \cup S_I$ vote for the Incumbent.  The Incumbent wins. 

If $v > \frac{1 + \alpha}{2 + \alpha} - \frac{1}{2}$ (the budget is large or the value of the public good is small), the Incumbent and the Challenger offer the public good to the same set of voters $S_I = S_C = [\frac{1}{2}, \frac{1 + \alpha}{2 + \alpha}]$. Voters with positions $t < \frac{1}{2}$ vote for the Challenger, and voters with positions $t > \frac{1}{2}$ vote for the Incumbent. The Incumbent wins. 

\end{proposition}
All proofs are in the Appendix.

In the first case presented in Proposition \ref{informed}, the budget is sufficiently small or the importance of the public good for the voters is sufficiently high so that the Incumbent can deplete the Challenger's resources by targeting the Challenger's moderate left-leaning supporters. More specifically, in all equilibria the Incumbent promises to spend the whole budget $v$ on moderate left, $S_I \subset [\frac{1 - \alpha}{2 - \alpha},\frac{1}{2}]$, and the Challenger - on moderate right voters, $S_C \subset [\frac{1}{2}, \frac{1 + \alpha}{2 + \alpha}]$. To win, the Challenger would have to spend more than the budget, so he has no profitable deviation.  The Incumbent would lose if she promised to spend less, so she also has no profitable deviation. Importantly, the multiplicity of equilibria does not pose a problem for the voters' welfare comparison between the informed and uninformed politicians scenario.    
%Equilibrium multiplicity does not affect the main result, which compares voters' welfare when politicians are or are not informed.  

In the second case presented in Proposition \ref{informed}, either the budget is sufficiently large or the importance of the public good is sufficiently low so that the Incumbent retains all of her core supporters and wins. 
%As the public good is unimportant for voters, it is hard for the Challenger to make the extreme right voters switch to support him. 
As the budget is sufficiently large, the Challenger offers the public good to the whole set of moderate Incumbent's supporters, $t\in[\frac{1}{2},\frac{1 + \alpha}{2 + \alpha}]$. In response, the Incumbent offers the public good to the same set of voters. It turns out that offering the public good to all moderate right-leaning voters is enough for the Incumbent to win. Indeed, because these voters are ideologically closer to the Incumbent, they vote for her even if the Challenger also offers them the public good. Voters with positions to the right of $\frac{1 + \alpha}{2 + \alpha}$ are "partisans" who always vote for the Incumbent because of ideology. As a result, the Incumbent gets half of the votes and wins regardless of what the Challenger does.  As the Challenger cannot win, he has no profitable deviation. To see why the Incumbent has no profitable deviation, consider a situation when she offers no public goods. In that case, to win, she needs extra $\frac{1 + \alpha}{2 + \alpha} - \frac{1}{2}$ votes. Hence, the Incumbent cannot win if she deviates to spending less or targeting some other subset of voters. 

Finally, demonstrating that there are no other equilibria  is more involved. However, the proof boils down to showing that in all other strategy profiles, either the winning politician can spend less, or the loser can win by copying the winner's move and selecting an additional small set of voters. 

It is worth noting that the implications for the case of a sufficiently large budget align well with empirical evidence. The model concludes that when politicians have abundant resources, the public goods are allocated to relatively richer voters. The empirical evidence suggests that in developed democracies, like the US and Western European countries, the policy outcomes tend to reflect the preferences of the relatively richer voters \citep{gilens2005inequality,bartels2016unequal,schakel2021unequal}.

The next Proposition characterizes the equilibria when politicians are uninformed. 

\begin{proposition}
\label{uninformed}
Suppose that politicians are uninformed. In the unique equilibrium
\begin{equation}
s_I = \frac{(2 - \alpha ) v}{\alpha  (1 - 2 v) + 2}
\end{equation}
and $s_C = v$. The Incumbent gets $1/2$ of votes and wins.
\end{proposition}
To see the intuition, observe that $1 - \frac{1+\alpha}{2 + \alpha}  >  \frac{1-\alpha}{2 - \alpha} $, so the Incumbent has more partisan supporters than the Challenger. Therefore, she can always win by choosing a sufficiently high $s_I$. In equilibrium, the Incumbent gets exactly $\frac{1}{2}$ of votes, because, otherwise, she could deviate by spending slightly less. On the other hand, the Incumbent's spending $s_I$ must be high enough so that she wins even if the Challenger chooses $s_C = v$. Otherwise, the Challenger could profitably deviate by choosing $s_C$ slightly smaller than $v$ and win. The result follows by direct computation. 

The following Proposition presents the main result, which is the welfare comparison between the cases of informed and uninformed politicians. 

\begin{proposition}
\label{welfare}
If $v>\frac{-\alpha ^3-6 \alpha ^2-8 \alpha }{2 \alpha ^3-16 \alpha -32}$ (the budget is large or the value of the public good is small), the social welfare of voters is higher when politicians are uninformed. Otherwise, the social welfare is higher when politicians are informed. 
\end{proposition}

By Proposition \ref{informed}, if politicians are informed, the budget is large, and the value of the public good is small, the Incumbent wins by targeting moderate right-leaning voters on the segment $[\frac{1}{2},\frac{1 + \alpha}{2 + \alpha}]$. Hence, her expenditures stay fixed even if the budget grows, and spending goes to relatively rich voters. By contrast, if politicians are uninformed and the budget increases, the Challenger can promise more, and the Incumbent has to match her offer. Also, because the public good is randomly allocated, some of the poor voters benefit from it. Therefore, the scenario with uninformed politicians leads to higher welfare of voters due to higher spending and the fact that poor voters may get the public good. 

If the budget is small or the value of the public good is large and the politicians are uninformed, then the public good goes at random to all voters including those who put little value on it. By contrast, informed politicians offer it to moderate voters who value it relatively strongly. As a result, the scenario with informed politicians leads to higher welfare. Figure \ref{fig:welfare}  illustrates Proposition \ref{welfare}.

\begin{figure}[H]
  \centering
  \includegraphics[width=\linewidth]{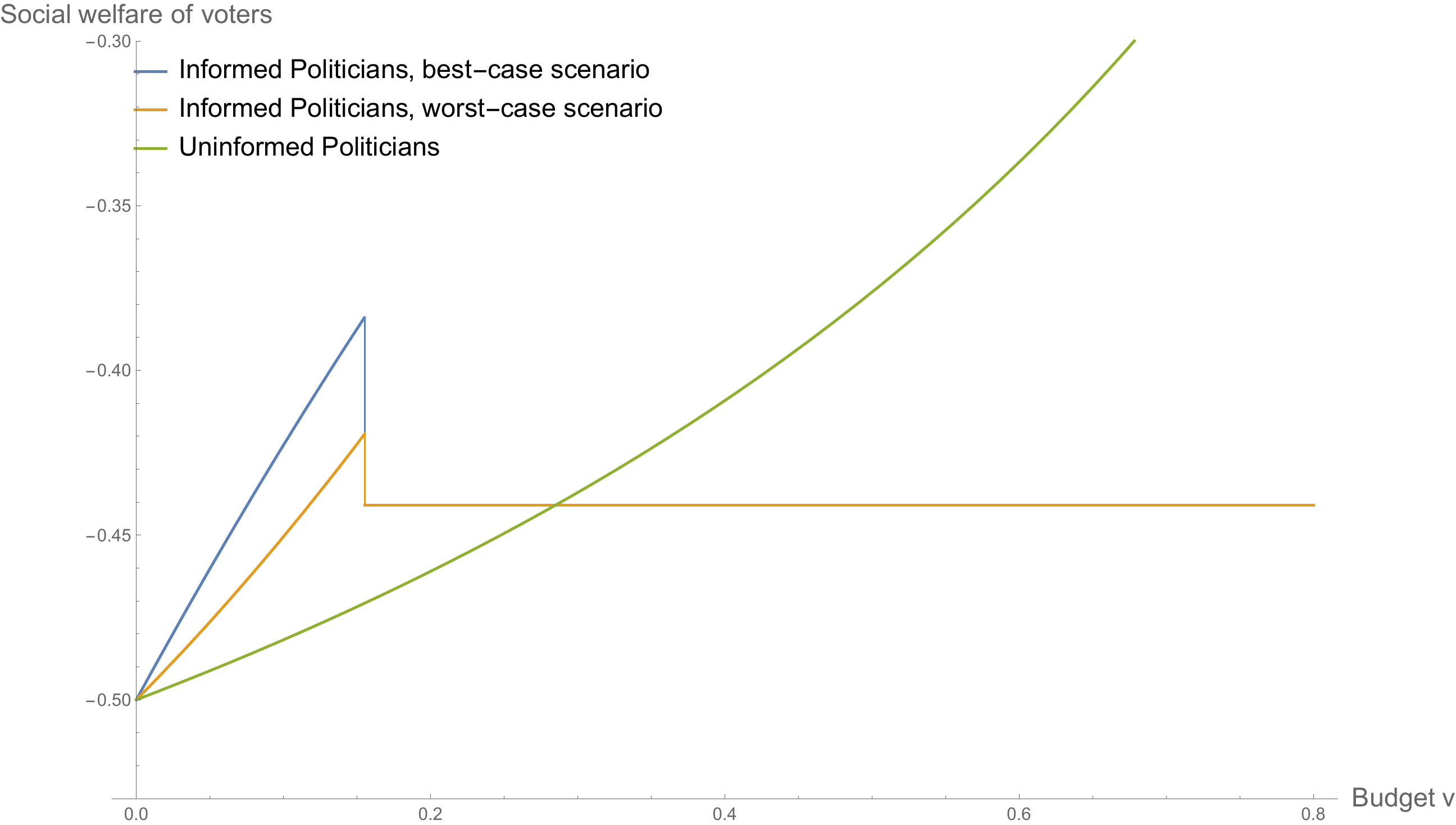}
  \caption{Social welfare of voters, value of the public good $\alpha $ fixed at $0.9$. }
  \label{fig:welfare}
\end{figure}

In the Figure, we compute the social welfare of voters under the worst- and best-case scenarios for the voters when politicians are informed (see the discussion in the proof of Proposition \ref{welfare} in Appendix). The worst-case scenario corresponds to such equilibrium that the informed Incumbent targets the richest from the moderate left-leaning voters,  $S_I = \left[\frac{1}{2}-v,\frac{1}{2}\right]$. Conversely, the best-case scenario arises when the Incumbent focuses on the poorest of the moderate left-leaning voters, where $S_I = \left[\frac{1-\alpha}{2-\alpha},\frac{1-\alpha}{2-\alpha}+v\right]$.  The discontinuity arises at a point where $v = \frac{1 + \alpha}{2 + \alpha} - \frac{1}{2} \approx 0.16$. As Proposition \ref{informed} implies, the discontinuity stems from the shift in the equilibrium strategy adopted by the politicians. Interestingly, the welfare of voters decreases at this point as the budget increases. The reason is that the Incumbent starts providing the public good to moderate-right voters who value it less instead of moderate-left voters who are poor and value it more.

\section{Discussion} \label{sect-concl}

The model compares two scenarios of political competition. In one, politicians know voter characteristics and can target public goods provision to them. In the other, they do not distinguish between voters and only decide on the level of public good spending. While these are extreme scenarios, they provide valuable benchmarks to analyze the impact of voter information on distributional decisions. The model's main contribution is that, unlike in a setting with a benevolent social planner, informed politicians may target swing voters who put low value on the public good. As a result, when the politicians' budget is large or the importance of public good is low, giving voter information to politicians leads to a lower welfare outcome as compared to when politicians are uninformed. The converse holds when the politicians' budget is small or the importance of public good is high. 

While the model analyzes public goods provision, a similar one could pertain to taxation. Consider a scenario where voters are distributed according to their income, with wealthier voters being more right-wing and vice versa. The level of government expenditures is fixed. However, politicians compete by offering tax cuts to voters, with poorer ones benefiting more intensely because of the decreasing marginal utility of money. In such a model, the same forces as in ours would lead to tax cuts being limited to a small group of swing voters, leading to social inefficiency. 

The model involves two strong assumptions. First, the budget is exogenous and fixed. This approach makes the analysis tractable and highlights non-obvious trade-offs relevant to providing politicians with voter information. Developing a model with heterogeneous voters and politicians deciding both on taxes and allocation of public goods may be a complicated task, but it is an interesting direction for further research. Second, in the model, the part of the budget that the winning politician does not spend on the public good does not benefit voters. This assumption is realistic because politicians often spend public money on benefits to lobby groups, which does not meet voters' needs. However, in reality, they may also spend less than the budget to prepare for unexpected future challenges. The government's intertemporal choice of spending does not relate to the objectives of this paper but is an interesting separate topic that deserves further attention. 

\bibliographystyle{plainnat}
\bibliography{refs_mod}

\appendix
\section*{Appendix}
\addcontentsline{toc}{section}{Appendices}
\renewcommand{\thesubsection}{\Alph{subsection}}

\subsection{Proofs} \label{app-proofs}

\begin{proof}[Proof of Proposition \ref{informed}]
We begin by characterizing "partisan voters" that do not respond to the provision of the public good and "swing voters" that do. Simple algebra shows that
\begin{equation}
\label{partisans1}
-(1 - t) > -t + \alpha(1 - t) \Leftrightarrow \frac{1 + \alpha }{2 + \alpha}<t<1
\end{equation}
and 
\begin{equation}
\label{partisans2}
-t > -(1 - t) + \alpha(1 - t)   \Leftrightarrow 0<t<\frac{1 - \alpha }{2 - \alpha }.
\end{equation}
Therefore, a voter with a position at $t > \frac{1 + \alpha }{2 + \alpha }$ ($t<\frac{1 - \alpha}{2 - \alpha}$) votes for the Incumbent (Challenger) even if only the Challenger (Incumbent) offers her the public good. For politicians, targeting them is wasteful because doing so makes them spend the budget but does not generate additional votes. Formally, a strategy profile where for some politician $i \in \{I,C\}$ $S_i \cap  \big( (0, \frac{1 - \alpha }{2 - \alpha}) \cup (\frac{1 + \alpha }{2 + \alpha },1) \big)\neq \emptyset$  can not be a part of an equilibrium. 

It is easy to show that any voter $t \in (\frac{1 - \alpha}{2 - \alpha}, \frac{1}{2})$ votes for the Incumbent if only the Incumbent offers her the public good and votes for the Challenger otherwise. Similarly, any voter $t \in (\frac{1}{2},\frac{1 + \alpha}{2 + \alpha})$ votes for the Challenger if only the Challenger offers her the public good and votes for the Incumbent otherwise. It is instructive to think of such voters as "left-" and "right-leaning swing voters," respectively. 

In equilibrium, politicians will only offer the public good to swing voters. It is
natural to separate the sets of voters $S_I$ and $S_C$ to which politicians offer the public good into subsets that consist of left- and right-leaning swing voters. Formally, define

\begin{align*}
X_I &\equiv S_I \cap  (\frac{1 - \alpha}{2 - \alpha}, \frac{1}{2}), 
&Y_I \equiv S_I \cap  (\frac{1}{2}, \frac{1 + \alpha}{2 + \alpha}),\\
X_C &\equiv S_C \cap  (\frac{1 - \alpha}{2 - \alpha}, \frac{1}{2}),
&Y_C \equiv S_C \cap   (\frac{1}{2}, \frac{1 + \alpha}{2 + \alpha}).
\end{align*}

We will now characterize two equilibria that exist under different parameter values. Later, we will show that other equilibria do not exist. First, we consider the case where the budget is large or the value of the public good is
small. Suppose that 
\begin{equation*}
v > \frac{1 + \alpha}{2 + \alpha} - \frac{1}{2}. 
\end{equation*}
Consider a strategy profile such that 
\begin{equation*}
S_I = Y_I =  S_C = Y_C = (\frac{1}{2}, \frac{1 + \alpha}{2 + \alpha}).
\end{equation*}
In this case, voters in $Y_I$ vote for the Incumbent even though the Challenger also offers them the public good. Voters to the right of $\frac{1 + \alpha}{2 + \alpha}$ always vote for the Incumbent. Hence, she gets exactly $\frac{1}{2}$ of votes and wins. The same is true if the Challenger chooses any $S_C$, so the Challenger has no profitable deviation. In substantive terms, the Incumbent wins by retaining the swing voters leaning in her favor. We will now show that the Incumbent also has no profitable deviation. Because $v > \frac{1 + \alpha}{2 + \alpha} - \frac{1}{2}$, she gets a strictly positive payoff. Suppose that the Incumbent deviates and chooses some other sets $X_I', Y_I'$. Then she wins only if 

\begin{equation*}
\begin{aligned}
\frac{1}{2} + \mu(X_I') + \mu(Y_I') - \mu(Y_C)  \geq \frac{1}{2} \Leftrightarrow \mu(X_I') + \mu(Y_I') \geq \mu(Y_C) = \frac{1 + \alpha}{2 + \alpha} - \frac{1}{2}.
\end{aligned}
\end{equation*}
It follows that the Incumbent cannot win by spending strictly less than she does in equilibrium and does not have a profitable deviation. Hence, this strategy profile is an equilibrium.

Consider now the case such that the budget is low or the value of the public good is large. Suppose that
\begin{equation*}
v < \frac{1 + \alpha}{2 + \alpha} - \frac{1}{2} 
\end{equation*} 
and consider a strategy profile such that
\begin{equation*}
S_I = X_I \subset (\frac{1 - \alpha}{2 - \alpha}, \frac{1}{2}), \mu(S_I) = v
\end{equation*}
and 
\begin{equation*}
S_C = Y_C \subset (\frac{1}{2}, \frac{1 + \alpha}{2 + \alpha}), \mu(S_C) = v.
\end{equation*}
First, observe that because $\frac{1}{2}-\frac{1-\alpha }{2-\alpha }>\frac{1 + \alpha}{2 + \alpha }-\frac{1}{2} > v$, this strategy profile is feasible and the winner pays weakly less than $v$.  The vote share for each politician is $\frac{1}{2}$, so the Incumbent wins. If the Incumbent deviates, then she wins only if 

\begin{equation*}
\begin{aligned}
\frac{1}{2} + \mu(X_I') + \mu(Y_I') - \mu(Y_C)   = \frac{1}{2} + \mu(X_I') + \mu(Y_I') - v \geq \frac{1}{2}\\ \Leftrightarrow 
\mu(X_I') + \mu(Y_I') \geq v
\end{aligned}
\end{equation*}
and therefore, the Incumbent cannot win and spend strictly less than $v$. Analogously, suppose that the Challenger deviates to some $X_C', Y_C'$. Then he wins only if 

\begin{equation*}
\begin{aligned}
\frac{1}{2} + \mu(X_C') + \mu(Y_C') - \mu(X_I)   = \frac{1}{2} +  \mu(X_C') + \mu(Y_C')  - v \geq \frac{1}{2}\\ \Leftrightarrow 
\mu(X_C') + \mu(Y_C') \geq v.
\end{aligned}
\end{equation*}
Hence, the Challenger also cannot win and spend strictly less than $v$. It follows that neither politician has a profitable deviation and this strategy profile is an equilibrium.

We proceed with proving that no other equilibria exist. First, we show that there is no equilibrium in which the Challenger wins, and then show there are no other equilibria in which the Incumbent wins. Consider strategy profiles such that the Challenger wins. In such profiles, the Challenger gets strictly more than $\frac{1}{2}$ of votes. Therefore $\mu(Y_C \setminus Y_I) > 0 $. But then the Challenger can deviate by setting $Y_C' = Y_C \setminus I_\epsilon$ where $I_\epsilon \subset (Y_C \setminus Y_I) $ is an interval of length $\epsilon$ and $\epsilon$ is sufficiently small. The interval exists by continuity of the voter type space. In this case, the Challenger continues to get strictly more than $\frac{1}{2}$ of votes and spends strictly less than before. 

Here, we show by contradiction that if the budget is large or the public good importance is low, then in any equilibrium, the Incumbent targets all moderate right-leaning voters, and the Challenger does the same. 
Consider strategy profiles such that the Incumbent wins. Suppose  that $v >   \frac{1 + \alpha}{2 + \alpha} -\frac{1}{2}  $. Consider any strategy profile such that $\mu(Y_I) < \frac{1 + \alpha}{2 + \alpha} - \frac{1}{2}$. Observe that if $\mu(X_I) \geq v$, then the Incumbent may deviate to setting $S_I = Y_I = (\frac{1}{2},\frac{1 + \alpha}{2 + \alpha})$, which ensures that she wins and spends less than $v$. Therefore, $\mu(X_I) < v$. Because  $\mu(Y_I) <  \frac{1 + \alpha}{2 + \alpha} - \frac{1}{2}$, there exists $ \bar \epsilon >0$ such that for any $\epsilon \in (0, \bar \epsilon)$ there is an interval $I_\epsilon \subset (\frac{1}{2},\frac{1 + \alpha}{2 + \alpha})$ with $\mu(I_\epsilon)  = \epsilon$. If the Challenger sets $X_C = X_I$ and $Y_C = I_\epsilon$, he wins strictly more than $\frac{1}{2}$ of the votes and pays $\mu(X_I) + \mu(I_\epsilon) = \mu(X_C) + \epsilon $ which is strictly less than $v$ when $\epsilon$ is sufficiently small. Hence, the Challenger has a profitable deviation, and such a strategy profile is not an equilibrium in which the Incumbent wins.  It follows that in any equilibrium, the Incumbent chooses $Y_I$ such that $\mu(Y_I) =   \frac{1 + \alpha}{2 + \alpha} - \frac{1}{2} \Rightarrow Y_I = (\frac{1}{2},  \frac{1 + \alpha}{2 + \alpha} )$. Because for the Incumbent, setting $Y_I = (\frac{1}{2},  \frac{1 + \alpha}{2 + \alpha} )$ is sufficient to win, her strategy such that $\mu(X_I) > 0 $ and $Y_I = (\frac{1}{2},  \frac{1 + \alpha}{2 + \alpha} )$ cannot be an part of equilibrium.

Here, we prove the absence of any other equilibria in the case of a small budget or high public good importance. We proceed in two steps: first, we show that in any equilibrium, the Incumbent spends all budget on moderate left-leaning voters; second, we show that in any equilibrium, the Challenger spends all budget on moderate right-leaning voters.
Suppose that $v < \frac{1 + \alpha}{2 + \alpha} -  \frac{1}{2} $. Consider a strategy profile such that $\mu(X_I) < v$. Observe that $\mu(Y_I) \leq v < \frac{1 + \alpha}{2 + \alpha} -  \frac{1}{2} $ because otherwise, the Incumbent would be spending more than the budget. Hence, for a sufficiently small $\epsilon$, there exists a segment $I_\epsilon \subset (\frac{1}{2},\frac{1 + \alpha}{2 + \alpha}) \setminus Y_I$ such that $\mu(I_\epsilon) = \epsilon$. For this reason and because $\mu(X_I) < v$, the Challenger can choose $X_C = X_I$ and $Y_C = I_\epsilon$, winning strictly more than $\frac{1}{2}$ of votes and spending $\mu(X_I) + \epsilon < v$. Hence, the Challenger has a profitable deviation. 
The Incumbent will not choose $X_I$ such that $\mu(X_I) > v$. It follows that in equilibrium, the Incumbent sets $\mu(X_I) = v$. Clearly, she also sets $\mu(Y_I) = 0$. 

The Incumbent's vote share is $\frac{1}{2} - \mu(X_C)  - \mu(Y_C)  + v \geq \frac{1}{2} \Leftrightarrow \mu(X_C) + \mu(Y_C) \leq v$, where the inequality follows from the assumption that the Incumbent wins. It trivially follows that $\mu(X_C)\leq v,\mu(Y_C) \leq v$. 
Suppose that $\mu(Y_C) < v$. In this case, the Incumbent can deviate to $S_I = Y_I = Y_C$, obtain $\frac{1}{2}$ of votes at a cost of $\mu(Y_I) = \mu(Y_C) < v$, and win. Therefore, there exists a profitable deviation for the Incumbent. 
%wins while spending strictly less than $v$. Setting $S_I = Y_I = Y_C$ is sufficient for the Incumbent to win. Hence, if this strategy profile is an equilibrium, the Incumbent sets $\mu(X_I) = 0$. 
%But then the Challenger could choose $X_C = I_\epsilon \subset (\frac{1 - \alpha}{2 - \alpha},\frac{1}{2})$ such that $\mu(I_\epsilon) = \epsilon$ and $Y_C = Y_I$. If $\epsilon$ is sufficiently small, this strategy ensures that the Challenger wins and spends less than $v$. So, this strategy profile cannot be an equilibrium in which the Incumbent wins. 
Thus, if $v < \frac{1 + \alpha}{2 + \alpha} - \frac{1}{2}$, the necessary condition for an equilibrium is given by $\mu(Y_C) = v$ and $\mu(X_C) = 0$. Given the sufficient condition for equilibrium obtained in the first part of the proof, a strategy profile is an equilibrium, given $v < \frac{1 + \alpha}{2 + \alpha} - \frac{1}{2}$, if and only if $S_I = X_I$ with $\mu(X_I) = v$ and $S_C = Y_C$ with $\mu(Y_C) = v$. 
\end{proof}

\begin{proof}[Proof of Proposition \ref{uninformed}]
The argument in the proof of Proposition \ref{informed} shows that voters with positions $t \in (0,\frac{1 - \alpha}{2 - \alpha})$ never vote for the Incumbent, voters with $t \in (\frac{1 - \alpha}{2 - \alpha},\frac{1}{2})$ vote for the Incumbent only if she offers them the public good but the Challenger does not, voters with $t \in (\frac{1}{2}, \frac{1 + \alpha}{2 + \alpha})$ vote for the Incumbent, unless the Challenger offers them the public good and the Incumbent does not, and voters with $t \in (\frac{1 + \alpha}{2 + \alpha},1)$ always vote for the Incumbent. Hence, the Incumbent's vote share equals 
\begin{equation}
\label{incs_vote_share}
\begin{aligned}
\eta(s_I,s_C) \equiv \left(1-\frac{1 + \alpha }{2 + \alpha }\right)+\left(\frac{1 + \alpha }{2 + \alpha }-\frac{1}{2}\right) (1-s_C (1-s_I))\\+\left(\frac{1}{2}-\frac{1-\alpha }{2-\alpha }\right) (1-s_C) s_I.
\end{aligned}
\end{equation}
Clearly, $\eta(s_I,s_C)$ is continuous and increasing (decreasing) in $s_I$ ($s_C$). If $\eta(s_I,s_C) > \frac{1}{2}$ then the Incumbent could profitably deviate by decreasing $s_I$. Similarly, if $\eta(s_I,s_C) < \frac{1}{2}$ then the Challenger could profitably deviate by decreasing $s_C$. It follows that if there is an equilibrium, then $\eta(s_I,s_C) = \frac{1}{2}$. Consider such strategy profiles. It is easy to verify by direct computation that for any $s_I,s_C\in (0,1)$ $\frac{d \eta(s_I,s_C)}{d s_C} < 0$. Therefore, if $\eta(s_I,s_C) = \frac{1}{2}$ and $s_C < v$, the Challenger could profitably deviate by increasing $s_C$ by a small amount. It follows that in any equilibrium, $\eta(s_I,s_C) = \frac{1}{2}$ and $s_C = v$. Substituting into (\ref{incs_vote_share}) and rearranging yields 

\begin{equation*}
s_I = \frac{(2 - \alpha ) v}{\alpha  (1 - 2 v) + 2} < v.
\end{equation*}
The inequality follows from a direct computation. Because $\frac{d \eta(s_I,s_C)}{d s_I} > 0$, if the Incumbent decreases $s_I$, she loses and gets $0$ while under the current strategy profile, she wins and gets a strictly positive payoff. Hence, there is no profitable deviation for the Incumbent. If the Challenger deviates to lower $s_C$, he loses. If he deviates to higher $s_C$, she wins but spends more than the budget. Thus, there is no profitable deviation for the Challenger. It follows that the strategy profile is an equilibrium. 
%Because $\frac{d \eta(s_I,s_C)}{d s_C} < 0$ and $s_C = v$, the Challenger can either win by spending more than the budget or lose by choosing $s_C < v$. Hence, he also has no profitable deviation. 
\end{proof}

\begin{proof}[Proof of Proposition \ref{welfare}]
We can compute the social welfare of voters using Propositions \ref{informed} and \ref{uninformed}. If $v > \frac{1 + \alpha }{2 + \alpha }-\frac{1}{2}$ and politicians are informed, the social welfare of voters equals
\begin{equation*}
\label{eq:SW_informed_low_v}
\int_0^1 -(1-t)dt + \int_{\frac{1}{2}}^{\frac{1 + \alpha }{2 + \alpha }} \alpha (1-t)dt = \frac{\alpha ^3-16 \alpha -16}{8 (\alpha +2)^2}.
\end{equation*}
If $v < \frac{1 + \alpha }{2 + \alpha }-\frac{1}{2}$ and politicians are informed, the social welfare of voters equals
\begin{equation*}
\int_0^1 -(1-t)dt + \int_{S_I} \alpha (1-t)dt,
\end{equation*}
where $S_I \subset (\frac{1 - \alpha}{2 - \alpha},\frac{1}{2})$ and $\mu(S_I) = v$. Because the function $1-t$ decreases in $t$, 
\begin{equation*}
\int_{S_I} \alpha (1-t)dt \geq \int_{\frac{1}{2} - v}^{\frac{1}{2}} \alpha (1-t)dt 
\end{equation*}
and therefore 
\begin{equation}
\label{eq:SW_informed_high_v}
\begin{aligned}
\int_0^1 -(1-t)dt + \int_{S_I} \alpha (1-t)dt \geq \\
\int_0^1 -(1-t)dt + \int_{\frac{1}{2} - v}^{\frac{1}{2}} \alpha (1-t)dt = \frac{1}{2} \left(\alpha  v^2+\alpha  v-1\right).
\end{aligned}
\end{equation}

If politicians are uninformed, the Incumbent wins by promising $s_I = \frac{(2 - \alpha ) v}{ \alpha(1 - 2 v) + 2}$
and therefore, the social welfare of voters becomes
\begin{equation*}
\begin{aligned}
\int_0^1 -(1-t)dt + \frac{(2 - \alpha ) v}{ \alpha(1 - 2 v) + 2} \int_{0}^{1} \alpha (1-t)dt = \frac{\alpha +\alpha ^2 v-4 \alpha  v+2}{\alpha  (4 v-2)-4}.
\end{aligned}
\end{equation*}
Direct computation shows that if $v > \frac{1 + \alpha }{2 + \alpha }-\frac{1}{2}$, then
\begin{equation*}
\begin{aligned}
\frac{\alpha ^3-16 \alpha -16}{8 (\alpha +2)^2} < \frac{\alpha +\alpha ^2 v-4 \alpha  v+2}{\alpha  (4 v-2)-4} \\
\Leftrightarrow v > \frac{-\alpha ^3-6 \alpha ^2-8 \alpha }{2 \alpha ^3-16 \alpha -32}. 
\end{aligned}
\end{equation*}
It also shows that if $v < \frac{1 + \alpha }{2 + \alpha }-\frac{1}{2}$ then
\begin{equation*}
\begin{aligned}
\frac{1}{2} \left(\alpha  v^2+\alpha  v-1\right) > \frac{\alpha +\alpha ^2 v-4 \alpha  v+2}{\alpha  (4 v-2)-4}.
\end{aligned}
\end{equation*}
Hence, welfare spending when politicians are informed leads to higher social welfare of voters even in the worst-case scenario. 

Finally, we include the derivation of the best-case scenario and worst-case scenario voter welfare, used in Figure \ref{fig:welfare}.
The social welfare of voters in the worst-case scenario is given by the second line in (\ref{eq:SW_informed_high_v}).
%\begin{equation*}
%\int_0^1 -(1-t)dt + \int_{\frac{1}{2} - v}^{\frac{1}{2}} \alpha (1-t)dt = \frac{1}{2} \left(\alpha  v^2+\alpha  v-1\right).
%\end{equation*}
If $v < \frac{1 + \alpha}{2 + \alpha} - \frac{1}{2}$, the social welfare of voters in the best-case scenario equals
\begin{equation*}
\int_0^1 -(1-t)dt + \int_{\frac{1 - \alpha}{2 - \alpha} }^{\frac{1 - \alpha}{2 - \alpha} + v} \alpha (1-t)dt = \frac{\alpha ^2 \left(-v^2\right)+\alpha  \left(2 v^2-2 v-1\right)+2}{2 (\alpha -2)}.
\end{equation*}

\end{proof}

\subsection{Extension: uninformed politicians and voters} \label{app-extension}

Consider changing the "uninformed politicians" scenario from Section \ref{sect-uninfo-pol} as follows. After the politicians announce the mass of public good to be provided,  $s_i, i\in\{I,C\}$, the voters observe the politicians' announcements; yet, they do not observe if they will have access to the public good or not, and thus attach a probability $s_i$ to having access to public good if politician $i$ wins. Thus, the voters' payoffs are now given by
\begin{equation*}
\begin{cases}
		-(1-t) + \alpha (1 - t) \cdot s_I, & \text{if the Incumbent wins}\\
         -t + \alpha (1 - t) \cdot s_C , & \text{if the Challenger wins.}
		 \end{cases}
\end{equation*}

Following the proof of Proposition \ref{uninformed}, it is straightforward to show that the Challenger chooses $s_C = v$ in equilibrium.
Thus, the position of an indifferent voter as a function of $s_I$ is implicitly given by
\[	-(1-t) + \alpha (1 - t) \cdot s_I =
         -t + \alpha (1 - t) \cdot v,\]
which yields 
\[\hat{t}=\frac{1+\alpha(v-s_I)}{2+\alpha(v-s_I)}. \]
Note that $\hat{t}$ decreases in $s_I$. Thus, in the unique equilibrium, $s_I$ solves $\hat{t} = \frac{1}{2}$, which yields $s_I = v$. This proves the following result.
\begin{proposition} \label{pr:extension}
Suppose that politicians are uninformed and, before voting, the voters do not know if they will receive the public good or not. In the unique equilibrium  both $s_I = v$ and $s_C = v$.   
\end{proposition}
In equilibrium described in Proposition \ref{pr:extension}, all voters to the right of $t=\hat{t}$ vote for the Incumbent, and all those to the left vote for the Challenger. Thus, the Incumbent obtains $\frac{1}{2}$ of votes and wins. Hence, the social welfare of voters in the "uninformed politicians" scenario is now given by 
\begin{equation}
\label{eq:SW_noinf_extension}
\int_0^1 -(1-t)dt + \int_{0}^{1} \alpha \cdot v (1-t)dt = \frac{1}{2}(v\alpha -1).
\end{equation}
Comparing the social welfare in the "informed politicians" scenario and "uninformed politicians" scenario, it is straightforward to show that for any $\alpha$ there exists a unique cutoff value of the budget $v$. 
\begin{proposition} \label{pr:ext_welfare}
If $v>\frac{1+\alpha}{2+\alpha}-\frac{1}{2}$ (the budget is large and the value of the public good is small), the social welfare of voters is higher when politicians and voters are uninformed. Otherwise, the social welfare is higher when politicians are informed.     
\end{proposition}
This welfare comparison is presented in Figure \ref{fig:welfare-ext}.  Finally, as the cutoff value of the budget $v$ from Proposition \ref{pr:ext_welfare} is strictly  between $0$ and $1$,  the welfare result from Proposition \ref{welfare} still holds qualitatively.
\begin{figure}[H]
  \centering
  \includegraphics[width=\linewidth]{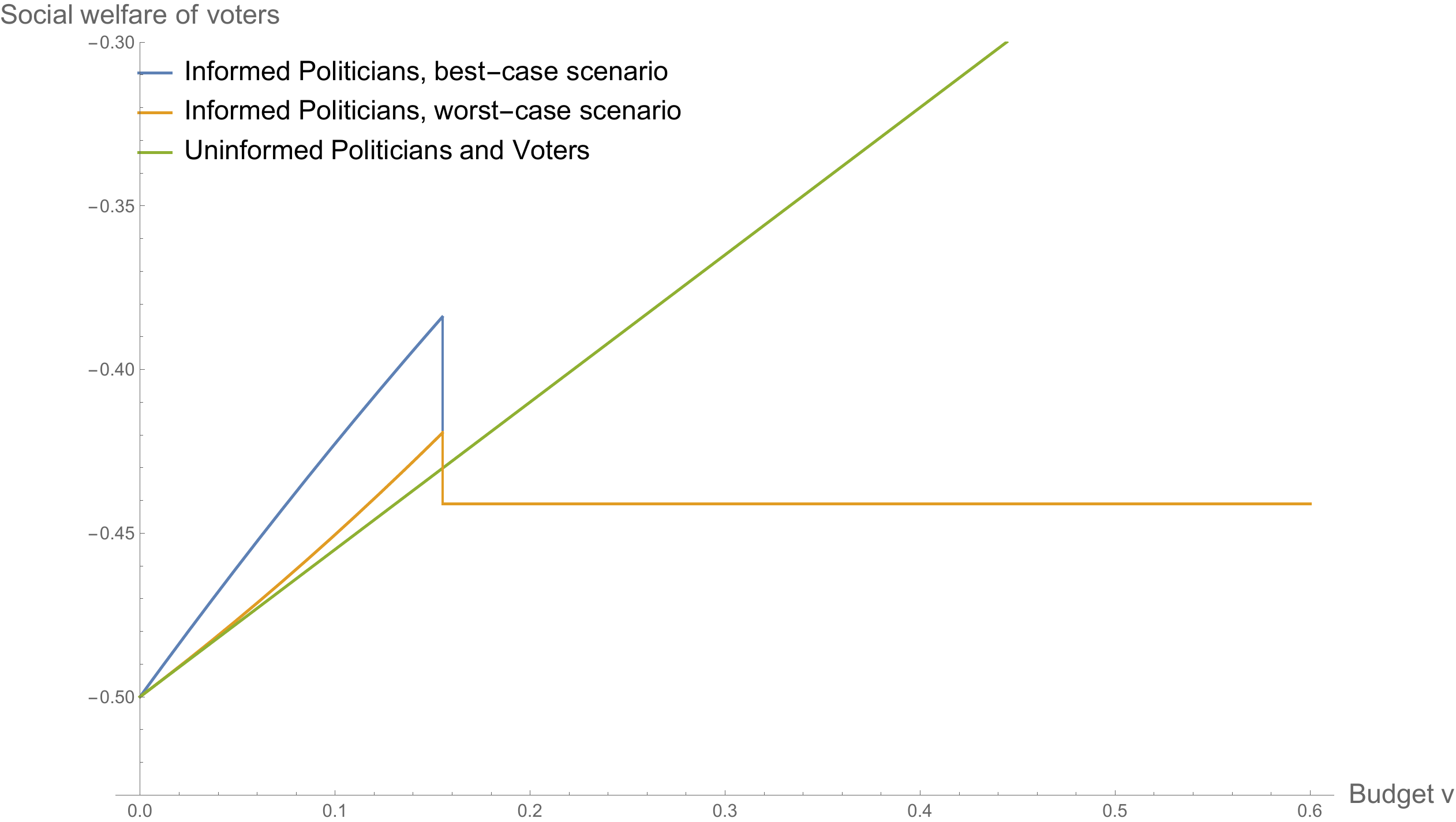}
  \caption{Social welfare of voters, value of the public good $\alpha $ fixed at $0.9$. }
  \label{fig:welfare-ext}
\end{figure}

%Equating the social welfare in the "informed politicians" scenario, given by (\ref{eq:SW_informed_low_v}) and (\ref{eq:SW_informed_high_v}), and the social welfare in the "uninformed politicians" scenario (\ref{eq:SW_noinf_extension}) yields the cutoff value of the budget $v$: \[v=\frac{1+\alpha}{2+\alpha}-\frac{1}{2}\in (0,1).\]  

\end{document}